\documentclass{article}%
\usepackage{amsmath}
\usepackage{amsfonts}
\usepackage{amssymb}
\usepackage{graphicx}%
\newtheorem{theorem}{Theorem}

\newtheorem{corollary}[theorem]{Corollary}

\newtheorem{remark}[theorem]{Remark}

\newenvironment{proof}[1][Proof]{\noindent\textbf{#1.}
}{\ \rule{0.5em}{0.5em}}
\begin{document}

\title{Spectral parameter power series for Sturm-Liouville problems}
\author{Vladislav V. Kravchenko and R. Michael Porter\\{\small Department of Mathematics, CINVESTAV del IPN, Unidad Quer\'{e}taro} \\{\small Libramiento Norponiente No. 2000, Fracc. Real de Juriquilla} \\{\small Queretaro, Qro. C.P. 76230 MEXICO } \\{\small e-mail: vkravchenko@qro.cinvestav.mx; mike@math.cinvestav.mx}}
\maketitle

\begin{abstract}
We consider a recently discovered representation for the general solution of
the Sturm-Liouville equation as a spectral parameter power series (SPPS). The
coefficients of the power series are given in terms of a particular solution
of the Sturm-Liouville equation with the zero spectral parameter. We show
that, among other possible applications, this provides a new and efficient
numerical method for solving initial value and boundary value problems.
Moreover, due to its convenient form the representation lends itself to
numerical solution of spectral Sturm-Liouville problems, effectively by
calculation of the roots of a polynomial. We discuss examples of the numerical
implementation of the SPPS method and show it to be equally applicable to a
wide class of singular Sturm-Liouville problems as well as to problems with
spectral parameter dependent boundary conditions.

\end{abstract}

\setcounter{section}{-1}

\section{Introduction}

In the recent work \cite{KrCV2008} a representation was obtained for solutions
of the equation%
\begin{equation}
(pu^{\prime})^{\prime}+qu=\lambda u \label{SL1}%
\end{equation}
in terms of a known non-trivial solution of the equation
\begin{equation}
(pu_{0}^{\prime})^{\prime}+qu_{0}=0 , \label{SL0}%
\end{equation}
where $p$, $q$, $u$, $u_{0}$ are complex-valued functions of the real variable
$x$ satisfying certain smoothness conditions, and $\lambda$ is an arbitrary
complex constant. The general solution of (\ref{SL1}) has the form of a power
series with respect to the spectral parameter $\lambda$. The fact that under
certain conditions $u$ is a complex analytic function of $\lambda$ has long
been known (see, e.g., \cite{Levitan}). The representation from
\cite{KrCV2008} gives a simple procedure for constructing its Taylor coefficients.

Our main purpose here is to show that this representation gives us, among
other possible applications, a simple and powerful method for numerical
solution of initial value, boundary value and spectral problems. We will
generalize slightly by considering the equation
\begin{equation}
(pu^{\prime})^{\prime}+qu=\lambda ru \label{SL}%
\end{equation}
and representing its general solution in terms of a solution of (\ref{SL0}),
which does not depend on the coefficient $r$.

In Section \ref{SectSL} we prove the main theorem and make several
observations which will apply to the subsequent applications. In Section
\ref{SectNSIVP} we observe that the rate of convergence of the spectral
parameter power series (SPPS) corresponding to (\ref{SL}) is quite easy to
estimate. This fact and the simple form of the solutions make the SPPS method
appropriate for solving initial and boundary value problems. Perhaps more
importantly, the representation of a solution of the Sturm-Liouville equation
in the form of an SPPS is well-suited for solving spectral problems as it
reduces problem to calculation of zeros of an analytic function defined by its
Taylor series. In Section \ref{SectSP} we show that even difficult examples
like the Coffey-Evans potential can be handled by the SPPS method; there is
little doubt that the numerical performance of the method can be significantly
improved using additional computational techniques discussed in the paper. In
Section \ref{SectSPDBC} we show that the appearance of the spectral parameter
in the boundary conditions does not provoke additional difficulties in
application of the SPPS method. As an example of application of the SPPS
method to singular Sturm-Liouville problems we consider in Section
\ref{SectSingProblems} a very interesting example of a highly non-self-adjoint
operator (with complex and singular coefficients) discussed in a number of
recent publications
\cite{Benilov,Boulton,Chugunova,Davies,Davies/Weir,Weir,Weir 2}. The SPPS
method makes it possible to obtain eigenvalues and eigenfunctions of the
problem even when other known algorithms encounter numerical difficulties.

\section{Solution of the Sturm-Liouville equation\label{SectSL}}

Here we formulate a generalization of the main result of \cite{KrCV2008},
which allows us to obtain a general solution of the Sturm-Liouville equation
(\ref{SL}) in the form of a spectral parameter power series. The proof does
not depend on pseudoanalytic function theory (which was the tool employed in
\cite{KrCV2008}).

\begin{theorem}
\label{ThSolSL}Assume that on a finite interval $[a,b]$, equation (\ref{SL0})
possesses a particular solution $u_{0}$ such that the functions $u_{0}^{2}r$
and $1/(u_{0}^{2}p)$ are continuous on $[a,b]$. Then the general solution of
(\ref{SL}) on $(a,b)$ has the form
\begin{equation}
u=c_{1}u_{1}+c_{2}u_{2} \label{genmain}%
\end{equation}
where $c_{1}$ and $c_{2}$ are arbitrary complex constants,
\begin{equation}
u_{1}=u_{0}%
{\displaystyle\sum\limits_{k=0}^{\infty}}
\lambda^{k}\widetilde{X}^{(2k)}\quad\text{and}\quad u_{2}=u_{0}%
{\displaystyle\sum\limits_{k=0}^{\infty}}
\lambda^{k}X^{(2k+1)} \label{gensol}%
\end{equation}
with $\widetilde{X}^{(n)}$ and $X^{(n)}$ being defined by the recursive
relations
\begin{equation}
\widetilde{X}^{(0)}\equiv1,\quad X^{(0)}\equiv1, \label{Xgen1}%
\end{equation}%
\begin{equation}
\widetilde{X}^{(n)}(x)=\left\{
\begin{tabular}
[c]{ll}%
$%
{\displaystyle\int\limits_{x_{0}}^{x}}
\widetilde{X}^{(n-1)}(s)u_{0}^{2}(s)r(s)\,ds$, & $n$ \text{odd,}\\
$%
{\displaystyle\int\limits_{x_{0}}^{x}}
\widetilde{X}^{(n-1)}(s)\frac{1}{u_{0}^{2}(s)p(s)}\,ds$, & $n$ \text{even,}%
\end{tabular}
\ \ \ \right.  \label{Xgen2}%
\end{equation}%
\begin{equation}
X^{(n)}(x)=\left\{
\begin{tabular}
[c]{ll}%
$%
{\displaystyle\int\limits_{x_{0}}^{x}}
X^{(n-1)}(s)\frac{1}{u_{0}^{2}(s)p(s)}\,ds$, & $n$ \text{odd,}\\
$%
{\displaystyle\int\limits_{x_{0}}^{x}}
X^{(n-1)}(s)u_{0}^{2}(s)r(s)\,ds$, & $n$ \text{even},
\end{tabular}
\ \ \ \ \ \ \right.  \label{Xgen3}%
\end{equation}
where $x_{0}$ is an arbitrary point in $[a,b]$ such that $p$ is continuous at
$x_{0}$ and $p(x_{0})\neq0$. Further, both series in (\ref{gensol}) converge
uniformly on $[a,b]$.
\end{theorem}

\begin{proof}
First we prove that $u_{1}$ and $u_{2}$ are indeed solutions of (\ref{SL})
whenever the application of the operator $L=\frac{d}{dx}p\frac{d}{dx}+q$ to
them makes sense. For this, note that if $Lu_{0}=0$, then $L$ can be written
in the factorized form $L=\frac{1}{u_{0}}\frac{d}{dx}\,p\,u_{0}^{2}\frac
{d}{dx}\frac{1}{u_{0}}$. Then application of $\frac{1}{r}L$ to $u_{1}$ gives%
\begin{align*}
\frac{1}{r}Lu_{1}  &  =\frac{1}{ru_{0}}\frac{d}{dx}\left(  pu_{0}^{2}\frac
{d}{dx}%
{\displaystyle\sum\limits_{k=0}^{\infty}}
\lambda^{k}\widetilde{X}^{(2k)}\right)  =\frac{1}{ru_{0}}\frac{d}{dx}%
{\displaystyle\sum\limits_{k=1}^{\infty}}
\lambda^{k}\widetilde{X}^{(2k-1)}\\
&  =u_{0}%
{\displaystyle\sum\limits_{k=1}^{\infty}}
\lambda^{k}\widetilde{X}^{(2k-2)}=\lambda u_{1}.
\end{align*}
In a similar way one can check that $u_{2}$ satisfies (\ref{SL}) as well. In
order to give sense to this chain of equalities it is sufficient to prove the
uniform convergence of the series involved in $u_{1}$ and $u_{2}$ as well as
of the series obtained by a term-wise differentiation. This can be done with
the aid of the Weierstrass M-test. Indeed, we have $\left\vert \widetilde
{X}^{(2k)}\right\vert \leq\left(  \max\left\vert ru_{0}^{2}\right\vert
\right)  ^{k}\left(  \max\left\vert \frac{1}{pu_{0}^{2}}\right\vert \right)
^{k}\frac{\left\vert b-a\right\vert ^{2k}}{\left(  2k\right)  !}$ and the
series $%
{\displaystyle\sum\limits_{k=0}^{\infty}}
\frac{c^{k}}{\left(  2k\right)  !}$ is convergent where
\begin{equation}
c=\left\vert \lambda\right\vert \left(  \max\left\vert ru_{0}^{2}\right\vert
\right)  \left(  \max\left\vert \frac{1}{pu_{0}^{2}}\right\vert \right)
\left\vert b-a\right\vert ^{2}. \label{c}%
\end{equation}
The uniform convergence of the series in $u_{2}$ as well as of the series of
derivatives can be shown similarly.

The last step is to verify that the Wronskian of $u_{1}$ and $u_{2}$ is
different from zero at least at one point (which necessarily implies the
linear independence of $u_{1}$ and $u_{2}$ on the whole segment $[a,b]$). It
is easy to see that by definition all the $\widetilde{X}^{(n)}(x_{0})$ and
$X^{(n)}(x_{0})$ vanish except for $\widetilde{X}^{(0)}(x_{0})$ and
$X^{(0)}(x_{0})$ which equal $1$. Thus%
\begin{equation}
u_{1}(x_{0})=u_{0}(x_{0}),\qquad u_{1}^{\prime}(x_{0})=u_{0}^{\prime}(x_{0}),
\label{at01}%
\end{equation}%
\begin{equation}
u_{2}(x_{0})=0,\qquad u_{2}^{\prime}(x_{0})=\frac{1}{u_{0}(x_{0})p(x_{0})}
\label{at02}%
\end{equation}
and the Wronskian of $u_{1}$ and $u_{2}$ at $x_{0}$ equals $1/p(x_{0})\not =0$.
\end{proof}

\begin{remark}
In the case $\lambda=0$, the solution (\ref{gensol}) becomes $u_{1}=u_{0}$ and
$u_{2}=u_{0}%
{\displaystyle\int\limits_{x_{0}}^{x}}
\frac{ds}{u_{0}^{2}(s)p(s)}$. The expression for $u_{2}$ is a well known
formula for constructing a second linearly independent solution.
\end{remark}

\begin{remark}
The result of Theorem \ref{ThSolSL} is valid for infinite intervals as well,
the series being uniformly convergent on any finite subinterval.
\end{remark}

\begin{remark}
\label{RemSing}One of the functions $ru_{0}^{2}$ or $1/(pu_{0}^{2})$ may not
be continuous on $[a,b]$ and yet\/ $u_{1}$ or $u_{2}$ may make sense. For
example, in the case of the Bessel equation $(xu^{\prime})^{\prime}-\frac
{1}{x}u=-\lambda xu,$ we can choose $u_{0}(x)=x/2.$ Then $1/(pu_{0}^{2})\notin
C[0,1]$. Nevertheless all integrals in (\ref{Xgen2}) exist and\/ $u_{1}$
coincides with the nonsingular $J_{1}(\sqrt{\lambda}x)$, while $u_{2}$ is a
singular solution of the Bessel equation.
\end{remark}

\begin{remark}
\label{RemRegular} In the regular case the existence and construction of the
required $u_{0}$ presents no difficulty. Let $p$ and $q$ be real valued,
$p(x)\neq0$ for all $x\in\lbrack a,b]$ and let\/ $p$, $p^{\prime}$, $r$ and
$q$ be continuous on $[a,b]$. Then (\ref{SL0}) possesses two linearly
independent regular solutions\/ $v_{1}$ and $v_{2}$ whose zeros alternate.
Thus one may choose $u_{0}=v_{1}+iv_{2}$.

Even when the coefficient functions are not real-valued, nonvanishing
solutions abound. Assume that the equation possesses a pair of linearly
independent solutions $v_{1}$ and $v_{2}$. Let $A_{j}=\{x\in(a,b)\colon
\ v_{j}(x)\not =0\}$, $j=1,2$. The function $\psi(c,x)=cv_{1}(x)+v_{2}(x)$ for
$c\in\mathbb{C}$, $x\in A_{1}$, has maximal (real) rank 2 (i.e., $\psi$ is a
submersion) if and only if $v_{1}(x)\not =0$. By the preimage theorem
\cite[chapter 1]{GP1974}, the preimage $\psi^{-1}(0)$ is a $1$-dimensional
submanifold of $\mathbb{C}\times A_{1}$ and therefore its projection
$B_{1}\subseteq\mathbb{C}$ has zero measure. For every $c\in\mathbb{C}%
\backslash B_{1}$, the linear combination $cv_{1}+v_{2}$ does not vanish in
$A_{1}$. The same reasoning gives us a null set\/ $B_{2}$ such that
$v_{1}+cv_{2}$ does not vanish in\/ $A_{2}$ when $c\in\mathbb{C}\backslash
B_{2}$. Altogether, there exist nonvanishing solutions throughout\/ $A_{1}\cup
A_{2}$, which is all of\/ $(a,b)$ because\/ $v_{1},$ $v_{2}$ are linearly
independent solutions of a second order differential equation.
\end{remark}

\begin{remark}
\label{RemLambda}The procedure for construction of solutions described in
Theorem \ref{ThSolSL} works not only when a solution is available for
$\lambda=0$, but in fact when a solution of the equation
\begin{equation}
(pu_{0}^{\prime})^{\prime}+qu_{0}=\lambda_{0}ru_{0} \label{eqlambda}%
\end{equation}
is known for some fixed $\lambda_{0}$. The solution (\ref{gensol}) now takes
the form%
\[
u_{1}=u_{0}%
{\displaystyle\sum\limits_{k=0}^{\infty}}
\left(  \lambda-\lambda_{0}\right)  ^{k}\widetilde{X}^{(2k)}\quad
\text{and}\quad u_{2}=u_{0}%
{\displaystyle\sum\limits_{k=0}^{\infty}}
\left(  \lambda-\lambda_{0}\right)  ^{k}X^{(2k+1)}.
\]
This can be easily verified by writing (\ref{SL}) as%
\[
\left(  L-\lambda_{0}r\right)  u=\left(  \lambda-\lambda_{0}\right)  ru.
\]
The operator on the left-hand side can be factorized exactly as in the proof
of the theorem, and the same reasoning carries through.
\end{remark}

\begin{remark}
\label{RemDouble} For calculating the series in (\ref{gensol}) it may be
convenient to calculate $X^{(n)}$ or $\widetilde{X}^{(n)}$ directly from
$X^{(n-2)}$ or $\widetilde{X}^{(n-2)}$. For example, when $n$ is even we have
\begin{align*}
\widetilde{X}^{(n)}(x)  & =\int_{x_{0}}^{x}\frac{1}{u_{0}(s)^{2}p(s)}%
\int_{x_{0}}^{s}u_{0}(t)^{2}r(t)\widetilde{X}^{(n-2)}(t)\,dt\,ds\\
& =\int_{x_{0}}^{x}(P(x)-P(t))u_{0}(t)^{2}r(t)\widetilde{X}^{(n-2)}(t)\,dt
\end{align*}
where $P^{\prime}=1/(u_{0}^{2}p)$.
\end{remark}

\begin{remark}
Other representations of the general solution of (\ref{SL}) as a formal power
series have been long known (see \cite[Theorem 1]{Trubowitz},
\cite{Chanane1998}) and used for studying qualitative properties of solutions.
The complicated manner in which the parameter $\lambda$ appears in those
representations makes that form of a general solution too difficult for
quantitative analysis of spectral and boundary value problems. In contrast,
the solution (\ref{genmain})-(\ref{Xgen3}) is a power series with respect to
$\lambda$, making it quite attractive for numerical solution of spectral,
initial value and boundary value problems.
\end{remark}

A special case of Theorem \ref{ThSolSL}, with $q\equiv0$, $\lambda=1$, was
known to H.\ Weyl.

\begin{corollary}
[\cite{Weyl}]\label{CorWeyl} Let\/ $1/p$ and\/ $r$ be continuous on\/ $[a,b]$.
The general solution of the equation
\begin{equation}
(pu^{\prime})^{\prime}=ru
\end{equation}
on $\left(  a,b\right)  $ has the form%
\begin{equation}
u=c_{1}u_{1}+c_{2}u_{2} \label{gensolSchr1}%
\end{equation}
where $c_{1}$ and $c_{2}$ are arbitrary constants and $u_{1}$, $u_{2}$ are
defined by (\ref{gensol})--(\ref{Xgen3}) with $u_{0}\equiv\lambda=1$.
\end{corollary}

This corollary enables us to find the particular solution $u_{0}$ discussed in
Remark \ref{RemRegular}.

\bigskip

\section{Numerical solution of initial value problems}

\label{SectNSIVP}

Consider the Sturm-Liouville equation (\ref{SL}) on $[a,b]$ with any desired
initial conditions. The numerical implementation of the solution via the
representation (\ref{gensol}) for a general solution is algorithmically
simple. One must consider the accuracy of calculation of the iterated
integrals in (\ref{Xgen2}) and (\ref{Xgen3}), and the rate of convergence of
the series (\ref{gensol}), because in numerical work one must work with
finitely many terms.

The main parameters that one can control are the number $M$ of subintervals in
which to divide $[a,b]$ when integrating numerically and the number $N$ of
powers in the truncated series. The relationship of $M$ to the accuracy of the
integrals is a standard question and will not be discussed here. In regards to
$N$, observe that one can not always expect a good approximation to $u$ over
all of $[a,b]$ with a series of $N$ terms, no matter how precisely the
integrals are calculated. However, using the estimate for $\left\vert
\widetilde{X}^{(2k)}\right\vert $ and $\left\vert X^{(2k-1)}\right\vert $ (see
the proof of Theorem \ref{ThSolSL} below) it is easy to obtain a rough but
useful estimate for the tail of the SPPS. Namely, consider $\left\vert
u_{1}-u_{1,N}\right\vert $ where $u_{1,N}=u_{0}%
{\displaystyle\sum\limits_{k=0}^{N}}
\lambda^{k}\widetilde{X}^{(2k)}$. We have
\begin{align*}
\left\vert u_{1}-u_{1,N}\right\vert  &  =\left\vert u_{0}\right\vert
\left\vert
{\displaystyle\sum\limits_{k=N+1}^{\infty}}
\lambda^{k}\widetilde{X}^{(2k)}\right\vert \leq\max\left\vert u_{0}%
\right\vert
{\displaystyle\sum\limits_{k=N+1}^{\infty}}
\frac{c^{k}}{\left(  2k\right)  !}\\
&  =\max\left\vert u_{0}\right\vert \left\vert \cosh\sqrt{c}-%
{\displaystyle\sum\limits_{k=0}^{N}}
\frac{c^{k}}{\left(  2k\right)  !}\right\vert
\end{align*}
where $c$ is defined by (\ref{c}). In a similar way one can see that the
remainder of the SPPS corresponding to $u_{2}$ is estimated by the tail of the
power series of $\sinh\sqrt{c}$. Thus, if a certain value of $N$ is seen to be
insufficient for achieving a required accuracy, the interval can be subdivided
and the initial value problem solved on the first subinterval. The initial
values of the solution for the second subinterval are calculated easily taking
into account that $u_{1}^{\prime}=\frac{u_{0}^{\prime}}{u_{0}}u_{1}+\frac
{1}{u_{0}p}%
{\displaystyle\sum\limits_{k=1}^{\infty}}
\lambda^{k}\widetilde{X}^{(2k-1)}$ (and analogously for $u_{2}$). Thus, no
numerical differentiation is necessary and this process can be continued with
little loss in accuracy.

The required particular solution $u_{0}$ may be calculated using any available
algorithm; in the examples presented below we have applied the formula of
Corollary \ref{CorWeyl}, applying the above subdivision procedure. All of the
calculations were performed with \textit{Mathematica} (Wolfram).

\section{Spectral problems}

\label{SectSP}

The fact that spectral Sturm-Liouville problems are related to the problem of
finding zeros of complex analytic functions of the variable $\lambda$ is quite
well known (see, e.g., \cite{Levitan}). For a regular Sturm-Liouville problem
the corresponding analytic function is even entire. The representation
(\ref{genmain})--(\ref{Xgen3}) allows us to obtain the Taylor series of that
analytic function explicitly. As an example, let us first consider a spectral
problem for (\ref{SL}) with the boundary conditions%
\begin{equation}
u(0)=0\quad\text{and\quad}u(1)=0. \label{bvpexample}%
\end{equation}
We suppose that the coefficients satisfy the conditions from Remark
\ref{RemRegular} and that $u_{0}$ is constructed as described there, taking
$x_{0}=0$. From the first boundary condition and (\ref{at01}), the constant
$c_{1}$ in (\ref{genmain}) must be zero. Then the spectral problem reduces to
finding values of $\lambda$ for which $u_{2}(1)=u_{0}(1)%
{\displaystyle\sum\limits_{k=0}^{\infty}}
\lambda^{k}X^{(2k+1)}(1)$ vanishes. In other words, this spectral problem
reduces to the calculation of zeros of the complex analytic function
$\kappa(\lambda)=%
{\displaystyle\sum\limits_{m=0}^{\infty}}
a_{m}\lambda^{m}$ where%
\[
a_{m}=u_{0}(1)X^{(2k+1)}(1).
\]

Now let $\alpha$ and $\beta$ be arbitrary real numbers and consider the more
general boundary conditions%
\begin{equation}
u(a)\cos\alpha+u^{\prime}(a)\sin\alpha=0\label{at0}%
\end{equation}%
\begin{equation}
u(b)\cos\beta+u^{\prime}(b)\sin\beta=0\label{atbeta}%
\end{equation}
together with equation (\ref{SL}). Taking the solutions $u_{1}$ and $u_{2}$
defined by (\ref{gensol}) and using (\ref{at01}), (\ref{at02}) with $x_{0}=a$,
we obtain from (\ref{at0}) the following equation,
\[
c_{1}(u_{0}(a)\cos\alpha+u_{0}^{\prime}(a)\sin\alpha)+c_{2}\frac{\sin\alpha
}{u_{0}(a)p(a)}=0,
\]
which gives $c_{2}=\gamma c_{1}$ when $\alpha\neq\pi n$, with $\gamma
=-u_{0}(a)p(a)(u_{0}(a)\cot\alpha+u_{0}^{\prime}(a))$, whereas $c_{1}=0$ when
$\alpha=\pi n.$ In the latter case the result is similar to the example
considered above, thus let us suppose $\alpha\neq\pi n$. From the definition
of $u_{1}$ and $u_{2}$ we have%
\[
u_{1}^{\prime}=\frac{u_{0}^{\prime}}{u_{0}}u_{1}+\frac{1}{u_{0}p}%
{\displaystyle\sum\limits_{k=1}^{\infty}}
\lambda^{k}\widetilde{X}^{(2k-1)}\quad\text{and\quad}u_{2}^{\prime}%
=\frac{u_{0}^{\prime}}{u_{0}}u_{2}+\frac{1}{u_{0}p}%
{\displaystyle\sum\limits_{k=0}^{\infty}}
\lambda^{k}X^{(2k)}.
\]
Then the boundary condition (\ref{atbeta}) implies that%
\[
\left(  u_{0}(b)\cos\beta+u_{0}^{\prime}(b)\sin\beta\right)  \left(
{\displaystyle\sum\limits_{k=0}^{\infty}}
\lambda^{k}\widetilde{X}^{(2k)}(b)+\gamma%
{\displaystyle\sum\limits_{k=0}^{\infty}}
\lambda^{k}X^{(2k+1)}(b)\right)
\]%
\[
+\frac{\sin\beta}{u_{0}(b)p(b)}\left(
{\displaystyle\sum\limits_{k=1}^{\infty}}
\lambda^{k}\widetilde{X}^{(2k-1)}(b)+\gamma%
{\displaystyle\sum\limits_{k=0}^{\infty}}
\lambda^{k}X^{(2k)}(b)\right)  =0.
\]
Thus the spectral problem (\ref{SL}), (\ref{at0}), (\ref{atbeta}) reduces to
the problem of calculating zeros of the analytic function $\kappa(\lambda)=%
{\displaystyle\sum\limits_{m=0}^{\infty}}
a_{m}\lambda^{m}$ where
\[
a_{0}=\left(  u_{0}(b)\cos\beta+u_{0}^{\prime}(b)\sin\beta\right)  (1+\gamma
X^{(1)}(b))+\frac{\gamma\sin\beta}{u_{0}(b)p(b)}%
\]
and%
\[
a_{m}=\left(  u_{0}(b)\cos\beta+u_{0}^{\prime}(b)\sin\beta\right)  \left(
\widetilde{X}^{(2m)}(b)+\gamma X^{(2m+1)}(b)\right)
\]%
\[
+\frac{\sin\beta}{u_{0}(b)p(b)}\left(  \widetilde{X}^{(2m-1)}(b)+\gamma
X^{(2m)}(b)\right)  ,\quad m=1,2,\ldots.
\]

This reduction of a Sturm-Liouville spectral problem lends itself to a simple
numerical implementation. To calculate the first $n$ eigenvalues we consider
the Taylor polynomial $\kappa_{N}(\lambda)=%
{\displaystyle\sum\limits_{m=0}^{N}}
a_{m}\lambda^{m}$ with $N\geq n$. Thus the numerical approximation of
eigenvalues of the Sturm-Liouville problem reduces to the calculation of zeros
of the polynomial $\kappa_{N}(\lambda)$.

There is no need to work with zeros of only one polynomial. It is well known
that in general the higher roots of a polynomial become less stable with
respect to small inaccuracies in coefficients. Our spectral parameter power
series method is well suited to overcome this problem and thus to calculate
higher eigenvalues with a good accuracy. This is done using Remark
\ref{RemLambda}. Suppose we have already calculated the eigenvalue
$\lambda_{0}$ using the procedure described above as a first root of the
obtained polynomial. Then for the next step we define $U_{0}=u_{1}+iu_{2}$
where $u_{1}$ and $u_{2}$ are defined by (\ref{gensol}) with $\lambda
=\lambda_{0}$. The function $U_{0}$ is then a solution of (\ref{eqlambda}). We
use it to obtain the eigenvalue $\lambda_{1}$ of the original problem
observing that $\lambda_{1}=\Lambda_{1}+\lambda_{0}$ where $\Lambda_{1}$ is
the first eigenvalue of the equation $(L-\lambda_{0}r)u=\Lambda u$ with the
same boundary conditions as in the original problem. This procedure can be
continued for calculating higher eigenvalues. Note that if $\lambda_{0}=0$ we
should begin this shifting procedure starting with $\lambda_{1}$.

Here we discuss some numerical examples.

\textit{Paine Problem.} A number of spectral problems which have become
standard test cases appear in \cite{PainedeHoogAnderssen,Pryce}. As a first
example we consider
\[
p(x)=-1,\quad q(x)=\frac{1}{(x+0.1)^{2}}.
\]%
\[
u(0)=0,\quad u(\pi)=0.
\]

The eigenvalues in the following table were calculated via SPPS using
integration on 10,000 subintervals for calculating $N=100$ powers of $\lambda
$. These eigenvalues were found as roots of a single polynomial (i.e., the
shifting of $\lambda$ as described in Remark \ref{RemLambda} was not applied).
Due to the sensitivity of the larger roots of the polynomial to errors in the
coefficients, 100-digit arithmetic was used.

\medskip

\begin{center}%
\begin{tabular}
[c]{|r|r@{.}l|r@{.}l|}\hline
$n$ & \multicolumn{2}{c|}{$\lambda_{n}$ \cite{Pryce}} &
\multicolumn{2}{c|}{$\lambda_{n}$ SPPS}\\\hline
0 & 1 & 5198658211 & 1 & 519865821099\\\hline
1 & 4 & 9433098221 & 4 & 943309822144\\\hline
2 & 10 & 284662645 & 10 & 28466264509\\\hline
3 & 17 & 559957746 & 17 & 55995774633\\\hline
4 & 26 & 782863158 & 26 & 78286315899\\\hline
5 & 37 & 964425862 & 37 & 96442587941\\\hline
6 & 51 & 113357757 & 51 & 11335707578\\\hline
7 & 66 & 236447704 & 66 & 23646092491\\\hline
8 & 83 & 338962374 & 83 & 33879073183\\\hline
9 & 102 & 42498840 & 102 & 4259718823\\\hline
10 & 123 & 49770680 & 123 & 512483827\\\hline
\end{tabular}

\end{center}

\medskip On the basis of the above values, a new calculation was made by
shifting with $\lambda^{\ast}=66$, resulting in the following improved
approximations for the last few eigenvalues.

\medskip

\begin{center}%
\begin{tabular}
[c]{|r|r@{.}l|r@{.}l|}\hline
$n$ & \multicolumn{2}{c|}{$\lambda_{n}$ \cite{Pryce}} &
\multicolumn{2}{c|}{$\lambda_{n}$ SPPS}\\\hline
7 & 66 & 236447704 & 66 & 23644770359\\\hline
8 & 83 & 338962374 & 83 & 33896237419\\\hline
9 & 102 & 42498840 & 102 & 42498839828\\\hline
10 & 123 & 49770680 & 123 & 49770680101\\\hline
11 & 146 & 55960608 & 146 & 55960605783\\\hline
12 & 171 & 61264485 & 171 & 61265439928\\\hline
\end{tabular}

\end{center}

\medskip With $\lambda^{*}=146$ and increasing the number of powers to
$N=150$, the following further values were obtained.

\begin{center}%
\begin{tabular}
[c]{|r|r@{.}l|r@{.}l|}\hline
$n$ & \multicolumn{2}{c|}{$\lambda_{n}$ \cite{Pryce}} &
\multicolumn{2}{c|}{$\lambda_{n}$ SPPS}\\\hline
11 & 146 & 55960608 & 146 & 55586199495330\\\hline
12 & 171 & 61264485 & 171 & 60875781110985\\\hline
13 & 198 & 65837500 & 198 & 65416389844202\\\hline
\end{tabular}

\end{center}

When the number of digits for internal calculations was increased to 150, SPPS
produced the same results.

\medskip\textit{Coffey-Evans equation.} This test case, defined by
\[
p(x)=-1,\quad q(x)=-2\beta\cos2x+\beta^{2}\sin^{2}2x.
\]%
\[
u(-\pi/2)=0,\quad u(\pi/2)=0,
\]
presents the challenge of distinguishing eigenvalues within the triple
clusters which form as the parameter $\beta$ increases. We present results for
$\beta=20,30,50$. In all cases given here the eigenvalues were obtained
without shifting $\lambda$.

\begin{center}
$\beta=20$. \newline$M=10,000$ subintervals, $N$= 180 powers, 100 digits of precision.%

\begin{tabular}
[c]{|l|r@{.}l|r@{.}l|}\hline
$n$ & \multicolumn{2}{c|}{$\lambda_{n}$ \cite{ChildChambers,Ledoux}} &
\multicolumn{2}{c|}{$\lambda_{n}$ SPPS}\\\hline
0 & -0 & 00000000000000 & 0 & 0000000000000003\\\hline
1 & 77 & 91619567714397 & 77 & 9161956771439703\\\hline
2 & 151 & 46277834645663 & 151 & 4627783464566396\\\hline
3 & 151 & 46322365765863 & 151 & 4632236576586490\\\hline
4 & 151 & 46366898835165 & 151 & 4636689883516575\\\hline
5 & 220 & 15422983525995 & 220 & 1542298352599497\\\hline
6 & 283 & 0948 & 283 & 0948146954014377\\\hline
7 & 283 & 2507 & 283 & 2507437431126800\\\hline
8 & 283 & 4087 & 283 & 4087354034293064\\\hline
\end{tabular}

\end{center}

\smallskip

\begin{center}
$\beta=30$\newline$M=10,000$ subintervals, $N=150$ powers, 100 digits  of
precision.
\begin{tabular}
[c]{|l|r@{.}l|r@{.}l|}\hline
$n$ & \multicolumn{2}{c|}{$\lambda_{n}$ \cite{Ledoux,Pryce}} &
\multicolumn{2}{c|}{$\lambda_{n}$ SPPS}\\\hline
0 & 0 & 00000000000000 & 0 & 000000000000000002\\\hline
1 & 117 & 946307662070 & 117 & 94630766206876\\\hline
2 & \multicolumn{2}{c|}{} & 231 & 664928928423790\\\hline
3 & 231 & 66492931296 & 231 & 664928928423791\\\hline
4 & \multicolumn{2}{c|}{} & 231 & 664930082035462\\\hline
5 & \multicolumn{2}{c|}{} & 340 & 888299091685489\\\hline
6 & \multicolumn{2}{c|}{} & 403 & 219684016171863\\\hline
7 & \multicolumn{2}{c|}{} & 403 & 219684016171917\\\hline
\end{tabular}

$\beta=50$ \newline$M=10,000$ subintervals, $N=150$ powers, 100 digits of
precision.
\begin{tabular}
[c]{|l|r@{.}l|r@{.}l|}\hline
$n$ & \multicolumn{2}{c|}{$\lambda_{n}$ \cite{Pryce}} &
\multicolumn{2}{c|}{$\lambda_{n}$ SPPS}\\\hline
0 & 0 & 00000000000000 & 0 & 000000000000000003\\\hline
1 & 197 & 968726516507 & 197 & 96872651650729\\\hline
2 & \multicolumn{2}{c|}{} & 391 & 807\\\hline
3 & 391 & 80819148905 & 391 & 810\\\hline
4 & \multicolumn{2}{c|}{} & 547 & 1397060\\\hline
\end{tabular}

\end{center}

\section{Sturm-Liouville problems with spectral parameter dependent boundary
conditions\label{SectSPDBC}}

In this section we consider Sturm-Liouville problems of the form%
\begin{equation}
(pu^{\prime})^{\prime}+qu = \lambda ru,\quad x\in\lbrack a,b], \label{SL51}%
\end{equation}
\begin{equation}
u(a)\cos\alpha+u^{\prime}(a)\sin\alpha=0, \quad\alpha\in\lbrack0,\pi),
\label{SL52}%
\end{equation}
\begin{equation}
\beta_{1}u(b)-\beta_{2}u^{\prime}(b)= \varphi(\lambda) \left(  \beta
_{1}^{\prime}u(b)-\beta_{2}^{\prime}u^{\prime}(b)\right)  , \label{SL53}%
\end{equation}
where $\varphi$ is a complex-valued function of the variable $\lambda$ and
$\beta_{1}$, $\beta_{2}$, $\beta_{1}^{\prime}$, $\beta_{2}^{\prime}$ are
complex numbers. This kind of problem arises in many physical applications (we
refer to \cite{BenAmara} and references therein) and has been studied in a
considerable number of publications
\cite{BenAmara,Chanane2008,CodeBrowne2005,CoskunBayram2005, Fulton77,Walter}.
For some special forms of the function $\varphi$ such as $\varphi
(\lambda)=\lambda$ or $\varphi(\lambda)=\lambda^{2}+c_{1}\lambda+c_{2}$,
results were obtained \cite{CodeBrowne2005}, \cite{Walter} concerning the
regularity of the problem (\ref{SL51})--(\ref{SL53}); we will not dwell upon
the details. Our purpose is to show the applicability of the spectral
parameter power series (SPPS) method to this type of Sturm-Liouville problems.
For simplicity, let us suppose that $\alpha=0$ and hence the condition
(\ref{SL52}) becomes $u(a)=0$. Then as was shown in the preceding section, if
an eigenfunction exists it necessarily coincides with $u_{2}$ up to a
multiplicative constant.

In this case condition (\ref{SL53}) becomes equivalent to the equality%
\begin{equation}
\left(  u_{0}(b)\varphi_{1}(\lambda)-u_{0}^{\prime}(b)\varphi_{2}%
(\lambda)\right)
{\displaystyle\sum\limits_{k=0}^{\infty}}
\lambda^{k}X^{(2k+1)}(b)-\frac{\varphi_{2}(\lambda)}{u_{0}(b)p(b)}%
{\displaystyle\sum\limits_{k=0}^{\infty}}
\lambda^{k}X^{(2k)}(b)=0 \label{eqSLparam}%
\end{equation}
where $\varphi_{1,2}(\lambda)=\beta_{1,2}-\beta_{1,2}^{\prime}\varphi
(\lambda)$. Calculation of eigenvalues given by (\ref{eqSLparam}) is
especially simple in the case of $\varphi$ being a polynomial of $\lambda$.
Precisely this particular situation was considered in all of the
abovementioned references concerning Sturm-Liouville problems with spectral
parameter dependent boundary conditions. For these problems the calculation of
eigenvalues using our method does not present any additional difficulty
compared to the parameter independent situation discussed in the preceding section.

\section{Singular problems\label{SectSingProblems}}

As was mentioned in Remark \ref{RemSing}, one of the solutions of the
Sturm-Liouville equation can be singular and nevertheless the method presented
here is still applicable. We show one such application to an interesting
problem first considered in \cite{Benilov} and then in a number of recent
publications \cite{Boulton,Chugunova,Davies,Davies/Weir,Weir,Weir 2}. We
consider on the interval $(-\pi,\pi)$ the singular non-symmetric differential
equation%
\begin{equation}
-i\varepsilon\frac{d}{dx}\left(  \sin x\frac{du}{dx}\right)  -i\frac{du}%
{dx}=\lambda u\label{HeatExample}%
\end{equation}
with $0<\varepsilon<2$ and periodic conditions at $-\pi$ and $\pi$. In spite
of the fact that (\ref{HeatExample}) is highly non-self-adjoint (with complex
and singular coefficients) all of the eigenvalues are real \cite{Weir} and the
spectrum is discrete \cite{Davies}. Several algorithms have appeared for
approximating the eigenvalues to which we compare our results. It is known
(e.g., \cite{Boulton}) that for each $\lambda\in\mathbb{C}$, equation
(\ref{HeatExample}) possesses a unique (up to scalar multiples) solution in
$L^{2}(-\pi,\pi)$ which we denote by $\varphi(x,\lambda)$. This solution may
be normalized by the condition $\varphi(0,\lambda)=1$. Any solution linearly
independent of $\varphi$ will blow up as $x\rightarrow0$. Also \cite{Boulton}
$\lambda$ is an eigenvalue if and only if
\[
\varphi(-\pi,\lambda)=\varphi(\pi,\lambda).
\]
It is not difficult to rewrite equation (\ref{HeatExample}) in the form
(\ref{SL}),
\begin{equation}
-i\varepsilon\frac{d}{dx}\left(  \sin x\left(  \tan\frac{x}{2}\right)
^{1/\varepsilon}\frac{du}{dx}\right)  =\lambda\left(  \tan\frac{x}{2}\right)
^{1/\varepsilon}u.\label{HeatExSL}%
\end{equation}
Choosing $u_{0}\equiv1$ as a particular solution corresponding to $\lambda=0$
we immediately find that $u_{1}$ defined by (\ref{gensol}) belongs to
$L^{2}(-\pi,\pi)$, while $u_{2}$ is the singular solution blowing up as
$x\rightarrow0$. Moreover, due to the particular form of the coefficients in
(\ref{HeatExSL}) we have from Remark \ref{RemDouble} that
\begin{align*}
\widetilde{X}^{(2k)}(x) &  =-\left(  \tan\frac{x}{2}\right)  ^{-1/\varepsilon
}\int_{0}^{x}\left(  \tan\frac{s}{2}\right)  ^{1/\varepsilon}\widetilde
{X}^{2(k-1)}(s)\,ds+\int_{0}^{x}\widetilde{X}^{2(k-1)}(s)\,ds,\quad\\
k &  =1,2,\ldots.
\end{align*}
Both integrals are well behaved at $x=\pi$. Moreover, the integral $\int
_{0}^{x}\left(  \tan\frac{s}{2}\right)  ^{1/\varepsilon}ds$ can be expressed
explicitly in terms of hypergeometric functions (produced for example by
symbolic manipulations in \textit{Mathematica} version 6). In the numerical
calculation of $\widetilde{X}^{(2k)}$ and subsequently of $u_{1}$ we may
approximate the first integral above by the sum of corresponding integrals on
subintervals, each one calculated as follows:
\[
\int_{x_{j-1}}^{x_{j}}\left(  \tan\frac{s}{2}\right)  ^{1/\varepsilon
}\widetilde{X}^{2(k-1)}(s)\,ds\approx\widetilde{X}^{2(k-1)}(\frac
{x_{j}+x_{j-1}}{2})\int_{x_{j-1}}^{x_{j}}\left(  \tan\frac{s}{2}\right)
^{1/\varepsilon}\,ds.
\]
We show the results of application of the SPPS method in comparison with some
values calculated previously. For $\varepsilon=0.5$ our values lie between
those previously published. For $\varepsilon=0.1$ our results display a
remarkable agreement with those of \cite{Davies}.

\begin{center}
Eigenvalues of (\ref{HeatExample}) for $\varepsilon=0.5$

\medskip%
\begin{tabular}
[c]{|l|r@{.}l|r@{.}l|r@{.}l|}\hline
n & \multicolumn{2}{c|}{$\lambda_{n}$ \cite{Weir 2}} &
\multicolumn{2}{c|}{$\lambda_{n}$ \cite{Chugunova}} &
\multicolumn{2}{c|}{$\lambda_{n}$ SPPS}\\\hline
1 & 1 & 16714 & 1 & 167342 & 1 & 16723\\\hline
2 & 2 & 96821 & 2 & 968852 & 2 & 96844\\\hline
3 & 5 & 48168 & 5 & 483680 & 5 & 48268\\\hline
4 & 8 & 71272 & 8 & 715534 & 8 & 71354\\\hline
5 & 12 & 66119 & \multicolumn{2}{c|}{} & 12 & 6618\\\hline
6 & 17 & 32643 & \multicolumn{2}{c|}{} & 17 & 3275\\\hline
7 & 22 & 71033 & \multicolumn{2}{c|}{} & 22 & 7110\\\hline
8 & 28 & 81106 & \multicolumn{2}{c|}{} & 28 & 8122\\\hline
9 & 35 & 62928 & \multicolumn{2}{c|}{} & 35 & 6311\\\hline
10 & 43 & 16666 & \multicolumn{2}{c|}{} & 43 & 1677\\\hline
\end{tabular}

\newpage\medskip Eigenvalues of (\ref{HeatExample}) for $\varepsilon=0.1$

\medskip%
\begin{tabular}
[c]{|l|r@{.}l|r@{.}l|r@{.}l|r@{.}l|}\hline
n & \multicolumn{2}{c|}{$\lambda_{n}$ \cite{Benilov}} &
\multicolumn{2}{c|}{$\lambda_{n}$ \cite{Weir 2}} &
\multicolumn{2}{c|}{$\lambda_{n}$ \cite{Davies}} &
\multicolumn{2}{c|}{$\lambda_{n}$ SPPS}\\\hline
1 & 1 & 0097 & 1 & 00940 & 1 & 00968 & 1 & 00968\\\hline
2 & 2 & 0733 & 2 & 07305 & 2 & 07334 & 2 & 07334\\\hline
3 & 3 & 2297 & 3 & 22894 & 3 & 22978 & 3 & 22978\\\hline
4 & 4 & 5012 & 4 & 50088 & 4 & 50134 & 4 & 50134\\\hline
5 & 5 & 8992 & 5 & 89968 & 5 & 89993 & 5 & 89993\\\hline
6 & 7 & 4298 & 7 & 43154 & 7 & 43194 & 7 & 43194\\\hline
7 & 9 & 0951 & 9 & 10034 & 9 & 10097 & 9 & 10097\\\hline
8 & 10 & 8945 & 10 & 90881 & 10 & 9092 & 10 & 9092\\\hline
9 & 12 & 8252 & 12 & 85742 & 12 & 8578 & 12 & 8578\\\hline
10 & 14 & 8820 & 14 & 94727 & 14 & 9478 & 14 & 9478\\\hline
15 & \multicolumn{2}{c|}{} & \multicolumn{2}{c|}{} & 27 & 5331 & 27 &
5331\\\hline
20 & \multicolumn{2}{c|}{} & \multicolumn{2}{c|}{} & 43 & 74 & 43 &
6923\\\hline
\end{tabular}

\end{center}

We also show the results we obtained for $\varepsilon=0.01$, for which we do
not know any previous result. For this reason we include information on the
accuracy of the eigenfunctions, estimated as follows. The calculated
eigenfunctions $u$ were effectively substituted into the differential equation
(\ref{HeatExSL}) by applying the corresponding integral operators which in
principle should produce $\lambda u$, and the discrepancy $\delta_{1}$ at the
right endpoint $\pi$ was tabulated. Then the discrepancy $\delta_{2}$ in the
right boundary condition was evaluated.

\begin{center}
Eigenvalues of (\ref{HeatExample}) for $\varepsilon=0.01$

\medskip%
\begin{tabular}
[c]{|l|r@{.}l|r@{.}l|r@{.}l|}\hline
n & \multicolumn{2}{c|}{$\lambda_{n}$ SPPS} & \multicolumn{2}{c|}{$\delta_{1}%
$} & \multicolumn{2}{c|}{$\delta_{2}$}\\\hline
1 & 1 & 0001 & 1 & 0$\times$10$^{-14}$ & 3 & 8$\times$10$^{-16}$\\\hline
2 & 2 & 0008 & 1 & 5$\times$10$^{-13}$ & 1 & 7$\times$10$^{-14}$\\\hline
3 & 3 & 00269 & 3 & 0$\times$10$^{-12}$ & 7 & 3$\times$10$^{-14}$\\\hline
4 & 4 & 00638 & 4 & 9$\times$10$^{-11}$ & 2 & 9$\times$10$^{-12}$\\\hline
5 & 5 & 01243 & 6 & 8$\times$10$^{-10}$ & 1 & 3$\times$10$^{-11}$\\\hline
6 & 6 & 02143 & 8 & 8$\times$10$^{-9}$ & 6 & 5$\times$10$^{-9}$\\\hline
7 & 7 & 03393 & 3 & 2$\times$10$^{-7}$ & 6 & 7$\times$10$^{-8}$\\\hline
8 & 8 & 05048 & 6 & 2$\times$10$^{-6}$ & 3 & 4$\times$10$^{-7}$\\\hline
9 & 9 & 07162 & 0 & 000090 & 4 & 0$\times$10$^{-6}$\\\hline
10 & 10 & 098 & 0 & 014 & 0 & 00027\\\hline
11 & 11 & 0223 & 10 & 8 & 0 & 0011\\\hline
\end{tabular}

\end{center}

In \cite{Davies/Weir} it was proved that the set of eigenvalues of the problem
tends to $\mathbb{Z}$ as $\varepsilon\rightarrow0$. Figure 1 (where the
eigenvalues have been calculated by SPPS) illustrates this assertion.%
\begin{figure}
[ptb]
\begin{center}
\includegraphics[
height=2.6152in,
width=4.2134in
]%
{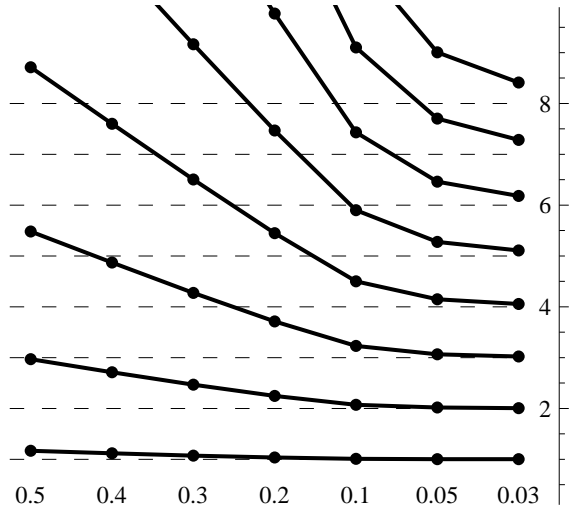}%
\caption{Convergence of $\lambda_{n}$ to $n$ as $\varepsilon\rightarrow0$.}%
\end{center}
\end{figure}

\section{Conclusions}

\smallskip In the present work a new representation for solutions of the
Sturm-Liouville equation is presented, which we call a spectral parameter
power series (SPPS). It gives a new natural and powerful method for solving
initial value, boundary value, and spectral problems. We have shown that it is
applicable not only to regular problems but also to singular problems as well
as to Sturm-Liouville problems with spectral parameter dependent boundary
conditions. The numerical realization of the method is simple, and its ideas
can be explained to undergraduate students in engineering, mathematics, and
physics. One of the important advantages of the SPPS method consists in the
automatic calculation of eigenfunctions together with the eigenvalues of a
spectral problem. The SPPS method is competitive with the best software for
numerical solution of ordinary differential equations and spectral problems
currently available. The authors' goal has been to show the broad
applicability and competitiveness of the method rather than to achieve
spectacular numerical results; we believe that much better performance of the
SPPS method can be achieved with the implementation of additional
computational techniques, some of which have been mentioned in this paper.

\textbf{Acknowledgements: }The authors express their gratitude to CONACYT for
partial support of this work as well as to Prof. Michael Levitin for
attracting our attention to the problem considered in section
\ref{SectSingProblems}.


\begin{thebibliography}{99}                                                                                               %


\bibitem {BenAmara}{\small Ben Amara J and Shkalikov A A 1999 A
Sturm-Liouville problem with physical and spectral parameters in boundary
conditions. Mathematical Notes 66, no. 2, 127--134.}

\bibitem {Benilov}{\small Benilov E S, O'Brien S B G and Sazonov I A 2003 A
new type of instability: explosive disturbances in a liquid film inside a
rotating horizontal cylinder. J. Fluid Mech. 497, 201--224. }

\bibitem {Boulton}{\small Boulton L, Levitin M and Marletta M 2008 A
PT-symmetric periodic problem with boundary and interior singularities.
arxiv:0801.0172v1.}

\bibitem {Chanane1998}{\small Chanane B 1998 Eigenvalues of Sturm-Liouville
problems using Fliess series. Applicable Analysis 69, 233--238.}

\bibitem {Chanane2008}{\small Chanane B 2008 Sturm-Liouville problems with
parameter dependent potential and boundary conditions. J. Comput. Appl. Math.
212 , no. 2, 282--290.}

\bibitem {ChildChambers}{\small Child M S and Chambers A V 1988 Persistent
accidental degeneracies for the Coffey-Evans potential. J. Phys. Chem 92,
3122--3124.}

\bibitem {Chugunova}{\small Chugunova M and Pelinovsky D 2007 Spectrum of a
non-self-adjoint operator associated with the periodic heat equation, preprint
http://arxiv.org/abs/math-ph/0702100v2.}

\bibitem {CodeBrowne2005}{\small Code W J and Browne P J 2005 Sturm-Liouville
problems with boundary conditions depending quadratically on the
eigenparameter. J. Math. Anal. Appl. 309, no. 2, 729--742.}

\bibitem {CoskunBayram2005}{\small Co\c{s}kun H and Bayram N 2005 Asymptotics
of eigenvalues for regular Sturm-Liouville problems with eigenvalue parameter
in the boundary condition. J. Math. Anal. Appl. 306, no. 2, 548--566. }

\bibitem {Davies}{\small Davies E B 2007 An indefinite  convection-diffusion
operator. LMS J. Comp. Math. 10, 288--306.}

\bibitem {Davies/Weir}{\small Davies E B and Weir J 2008 Convergence of
eigenvalues for a highly non-self-adjoint differential operator.
arxiv:0809.0787v1.}

\bibitem {Fulton77}{\small Fulton Ch T 1977 Two-point boundary value problems
with eigenvalue parameter contained in the boundary conditions. Proc. Roy.
Soc. Edinburgh Sect. A 77, no. 3--4, 293--308. }

\bibitem {GP1974}{\small Guillemin V and Pollack A 1974 Differential topology.
Prentice-Hall, Inc., Englewood Cliffs, N.J.}

\bibitem {KrCV2008}{\small Kravchenko V V 2008 A representation for solutions
of the Sturm-Liouville equation.\ Complex Variables and Elliptic Problems,
2008, v. 53, 775--789.}

\bibitem {Ledoux}{\small Ledoux V 2007 Study of Special Algorithms for
solving Sturm-Liouville and Schr\"{o}dinger Equations, thesis  Universiteit
Gent}

\bibitem {Levitan}{\small Levitan B M and Sargsjan I S 1991 Sturm-Liouville
and Dirac operators. Dordrecht: Kluwer Acad. Publ.}

\bibitem {Trubowitz}{\small P\"{o}schel J and Trubowitz E 1987 Inverse
spectral theory. Boston: Academic Press.}

\bibitem {PainedeHoogAnderssen}{\small Paine J W, De Hoog F R and Anderssen R
R 1981 Computing 26, 123--139}

\bibitem {Pryce}{\small Pryce J D 1993 Numerical solution of Sturm-Liouville
problems. Clarendon Press.}

\bibitem {Walter}{\small Walter J 1973 Regular eigenvalue problems with
eigenvalue parameter in the boundary condition. Math. Z. 133, 301--312. }

\bibitem {Weir}{\small Weir J 2008 An indefinite convection-diffusion operator
with real spectrum. Applied Mathematics Letters, in press.}

\bibitem {Weir 2}{\small Weir J 2008 Correspondence of the eigenvalues of a
non-self-adjoint operator to those of a self-adjoint operator.
arxiv:0801.4959v2.}

\bibitem {Weyl}{\small Weyl H 1910 \"{U}ber gew\"{o}hnliche
Differentialgleichungen mit Singularit\"{a}ten und die zugeh\"{o}rigen
Entwicklungen willk\"{u}rlicher Funktionen. (German) Math. Ann. 68, no. 2,
220--269.}
\end{thebibliography}
\end{document}